\documentclass[letterpaper, 10 pt, conference]{ieeeconf}  
\IEEEoverridecommandlockouts     
\newtheorem{mydef}{Definition}

\newtheorem{myprob}{Problem}

\newtheorem{mypro}{Proposition}
\newtheorem{myexm}{Example}

\newtheorem{remark}{Remark}

\usepackage{amsfonts} 
\DeclareMathSymbol{\shortminus}{\mathbin}{AMSa}{"39}

\usepackage{multirow}
\usepackage{makecell}
\usepackage{booktabs}
\usepackage{array}

\usepackage{graphicx}
\graphicspath{{./Pics/}}
\DeclareGraphicsExtensions{.pdf}
\usepackage{float}
\usepackage{amsmath}
\usepackage{amsfonts}
\usepackage{epsfig}
\usepackage{graphicx}
\usepackage{arydshln}
\usepackage{algpseudocode}
\usepackage{verbatim}
\usepackage{subfigure}
\usepackage{enumerate}
\usepackage{rotating}
\usepackage{color}
\usepackage{caption} 
\usepackage{mathrsfs}
\usepackage{amssymb}
\usepackage[linesnumbered,ruled,vlined]{algorithm2e}
\usepackage{cite}
\usepackage{float}
\usepackage{soul}
\usepackage{tikz}
\usepackage{hyperref}
\hypersetup{hidelinks} 
\usepackage{soul}
\usepackage{booktabs}
\usepackage{multirow}
\usepackage{makecell}
\usepackage{array}
\usetikzlibrary{arrows.meta}
\usetikzlibrary{patterns}
\usepackage{tcolorbox}
\usepackage{fancyhdr}

\floatname{algorithm}{Procedure}

\title{\LARGE \bf
Online Monitoring of Dynamic Systems for Signal Temporal Logic Specifications with Model Information}

\author{Xinyi Yu, Weijie Dong, Xiang Yin and Shaoyuan Li
\thanks{This work was supported by the National Key Research and Development Program of China (2018AAA0101700) and the National Natural Science Foundation of China (62061136004, 62173226, 61833012).}
	\thanks{Xinyi Yu, Weijie Dong, Xiang Yin and Shaoyuan Li are with Department of Automation and Key Laboratory of System Control and Information Processing, Shanghai Jiao Tong University, Shanghai 200240, China.
	{\tt\small $\{$yuxinyi-12, wjd\_dollar, yinxiang, syli$\}$@sjtu.edu.cn}. }%
}

\begin{document}
	
\maketitle
\thispagestyle{empty}
\pagestyle{empty}
\setlength{\abovecaptionskip}{0pt}
\setlength{\belowcaptionskip}{3pt}
\setlength{\textfloatsep}{0pt}

\begin{abstract}
Online monitoring aims to evaluate or to predict, at runtime, whether or not the behaviors of a system satisfy some desired specification. It plays a key role in safety-critical cyber-physical systems. In this work, we propose a new model-based approach for online monitoring for specifications described by Signal Temporal Logic (STL) formulae. Specifically, we assume that the observed state traces are generated by an underlying dynamic system whose model is known. The main idea is to consider the dynamic of the system when evaluating the satisfaction of the STL formulae. To this end, effective approaches for the computation of feasible sets for STL formulae are provided. We show that, by explicitly utilizing the model information of the dynamic system, the proposed online monitoring algorithm can falsify or certify of the specification in advance compared with existing algorithms, where no model information is used. We also demonstrate the proposed monitoring algorithm by case studies.
\end{abstract}

\section{Introduction}
Cyber-Physical Systems (CPS) are man-made modern engineering systems involving both computational devices and physical dynamics. Safety is one of the major considerations in the designs of many CPS such as intelligent transportation systems, smart manufacturing systems and medical devices. For those safety-critical systems, it is  crucial to determine whether or not the behaviors of the system satisfy some desired high-level specifications. For example, once we detect that the system has violated or will inevitablely violate the desired specification, additional corrective actions can be taken to ensure safety.

Specification-based monitoring is one of the major techniques in evaluating behavior correctness of CPS \cite{bartocci2018specification}.  In this context, it is usually assumed that the desired behavior of the system is described by a specification formula and the state traces (a.k.a.\ signals) generated by the system is observed by a \emph{monitor} that can issue alarms when the specification is violated. In the past years, numerous algorithms have been developed for monitoring specifications described by, e.g., Linear Temporal Logic (LTL) \cite{eisner2003reasoning}, Metric Temporal Logic (MTL) \cite{thati2005monitoring,dokhanchi2014line} and  Signal Temporal Logic (STL) \cite{donze2013efficient,deshmukh2017robust}.
Recent applications of specification-based monitoring techniques include, e.g.,  autonomous vehicles \cite{sahin2020autonomous} and smart cities \cite{ma2021novel}.   

Depending on what information can be utilized by the monitor, the monitoring problem can be categorized as \emph{offline} and \emph{online}. In offline monitoring, it is assumed that the complete signal to evaluate has already been generated and the monitor needs to determine either the Boolean satisfaction or the quantitative satisfaction degree of the complete signal. Such offline technique is usually used in the design phase to evaluate the simulated traces of the system prototype.  On the other hand, when the CPS is operating online, the monitor only observes \emph{partial} state trace that has been generated so far. Therefore, \emph{online monitoring} focuses on evaluating signals in real time during the operation of the system in order to, e.g., issue   alarms or to trigger corrective actions. 

In the context of online monitoring, the monitor may make the following evaluations on the observed partial signals: 
(i) the specification cannot be satisfied, i.e., there no future possibility to correct the signal;
(ii) the specification has already been satisfied, i.e.,  the future signal does not matter; or 
(iii) inconclusive, i.e., the signal can be either satisfied or not depending on what will happen in the future.  In the past years, numerous algorithms have been developed for online monitoring for specifications described by temporal logic formulae.  For example, the basic setting is to consider monitoring the Boolean satisfaction of LTL formulae \cite{bauer2011runtime,abate2019monitor,mascle2020ltl} or MTL formulae \cite{ho2014online}.  
In \cite{dokhanchi2014line, deshmukh2017robust}, algorithms have been developed for quantitatively monitoring the satisfaction of specifications by using robust semantics of STL formulae.

Most of the aforementioned online monitoring techniques are \emph{model-free} in the sense that the satisfaction of the specification is only evaluated based on the observed signal without considering the dynamic of the  system. In some cases, however, the model of the underlying system, when it is known, can provide additional information to accelerate monitoring process. For example, let us consider a scenario, where for an observed signal, a model-free monitor may provide inconclusive evaluation since the partial signal can be extended to either satisfiable  or unsatisfiable  signals. However, those satisfiable  continuations may not be feasible physically in the dynamic system. In this scenario, by leveraging the model information of the dynamic system, the monitor can better assert that the specification cannot be satisfied before it is actually violated.  

Motivated by the above observations, in this paper, we propose a new \emph{model-based} approach for online monitoring of dynamic systems. Specifically, we consider specifications described by a fragment of STL formulae, where  the horizons of different temporal operators have no overlap.  STL formulae are interpreted over continuous time signals and have the advantage of   quantitatively  evaluating the degree of the satisfaction or violation using robust semantics \cite{maler2004monitoring,lindemann2018control,gilpin2020smooth,salamati2021data,hashimoto2022stl2vec,lindemann2019robust}. 
The monitor aims to issue alarms when the specification has  already or will inevitably be violated. However, different from existing approaches, here we explicitly consider the model information of underlying dynamic system. Specifically, we consider a discrete-time nonlinear system. In order to incorporate the model information into the evaluation of STL formulae, we propose the notion of \emph{feasible sets}, which are the regions of states from which STL formulae can potential be satisfied considering the system dynamic. Effective algorithms have been developed for computing feasible sets offline. To monitor the specification in real-time, we propose online monitoring algorithms that correctly combine both the online observed partial signals and the offline computed feasible sets. We show that the proposed model-based monitoring algorithm may predict the violation of the specification in advance compared with existing model-free approaches. Hence, it may leave more time for the system to take corrective actions to ensure safety.

The rest of the paper is organized as follows. 
We present some basic preliminaries in Section~\ref{sec:pre} and formulate the problem in Section~\ref{sec:prob}.  Section~\ref{sec:online} present the main body of the  online monitoring algorithm, which uses feasible sets that are computed offline in  Section~\ref{sec:offline}.  The overall framework is demonstrated by   case studies in Section~\ref{sec:case} and finally, we conclude this work in Section~\ref{sec:con}.

\section{Preliminary}\label{sec:pre}
\subsection{System Model}

We consider a discrete-time control system of form
\begin{equation}\label{eq:system}
	x_{k+1} = f(x_k, u_k),
\end{equation}
where 
$x_k \in \mathcal{X}\subseteq \mathbb{R}^n$ is the state at time $k$, 
$u_k \in \mathcal{U}\subseteq \mathbb{R}^m$ is the control input at time $k$
and 
$f:\mathcal{X}\times \mathcal{U} \to \mathcal{X}$ is a dynamic function of the system.

Suppose that the system is in state $x_k\in  \mathcal{X}$ at time instant $k\in\mathbb{Z}_{\geq 0}$.  Then given a sequence of control inputs
$\mathbf{u}_{k:N-1}=  u_k u_{k+1} \dots u_{N-1}  \in \mathcal{U}^{N-k}$,
the  solution of the system is a sequence of states
$\xi_f(x_k,\mathbf{u}_{k:N-1}) = {\mathbf{x}_{k+1:N} = } x_{k+1} \dots x_{N}\in \mathcal{X}^{N-k}$ such that $x_{i+1}=f(x_i,u_i), i=k,\dots, N-1$.

\subsection{Signal Temporal Logic}
We use Signal Temporal Logic (STL) formulae  with bounded-time temporal operators \cite{maler2004monitoring} to describe whether or not the trajectory of the system satisfies some desired high-level properties. 
Formally, the syntax of STL formulae is as follows 
\[
\Phi ::=   \top\mid \pi^\mu \mid \neg \Phi \mid \Phi_1 \wedge \Phi_2 \mid \Phi_1 \mathbf{U}_{[a,b]} \Phi_2,
\]
where $\top$ is the \textsf{true} predicate, 
$\pi^\mu$ is an atomic predicate whose truth value is determined by the sign of its underlying predicate function $\mu:\mathbb{R}^n \to \mathbb{R}$  and it is true at state $x_k$ when $\mu(x_k) \geq 0$; otherwise it is false.
Notations $\neg$ and $\wedge$ are the standard Boolean operators ``negation" and ``conjunction", respectively,  which can further induce ``disjunction" by $\Phi_1 \vee \Phi_2:=\neg(\neg \Phi_1 \wedge \neg \Phi_2)$
and ``implication" by $\Phi_1 \to \Phi_2:= \neg \Phi_1 
\vee   \Phi_2$. 
$\mathbf{U}_{[a,b]}$ is the temporal operator ``\emph{until}", where $a,b\in \mathbb{N}$ are two time instants with $a \leq b$ 
and $[a,b]\subseteq \mathbb{N}$ denotes the set of all integers between
$a$ and $b$. 

STL formulae are evaluated on state sequence $\mathbf{x}=x_0x_1\cdots$.  
We use notation $(\mathbf{x},k) \models \Phi$ to denote that sequence $\mathbf{x}$ satisfies STL formula $\Phi$ at time instant $k$. 
The reader is referred to \cite{maler2004monitoring} for more details on the semantics of STL formulae. 
Particularly, we have 
$(\mathbf{x},k)\models \pi^\mu$ iff $\mu(x_k)\geq 0$, i.e., $\mu(x_k)$ is non-negative for the current state $x_k$,  
and 
$(\mathbf{x},k)\models \Phi_1 \mathbf{U}_{[a,b]} \Phi_2$ 
iff $\exists k' \in [k+a, k+b]$ such that $(\mathbf{x},k') \models \Phi_2$
and  $\forall k'' \in [k, k']$, we have $(\mathbf{x},k'') \models \Phi_1$, i.e., 
$\Phi_2$ will hold at some instant between $[k+a, k+b]$ in the future and before that $\Phi_1$ always holds. 
Furthermore, 
we can also induce temporal operators
\begin{itemize}
	\item 
	``\emph{eventually}" $\mathbf{F}_{[a,b]} \Phi:= \top \mathbf{U}_{[a,b]} \Phi$
	such that it holds when $(\mathbf{x},k) \models \Phi$ for some $k'\in [k+a,k+b]$; and 
	\item 
	``\emph{always}" $\mathbf{G}_{[a,b]} \Phi:=\neg \mathbf{F}_{[a,b]} \neg \Phi$ such that it holds  when $(\mathbf{x},k) \models \Phi$ for any $k'\in [k+a,k+b]$. 
\end{itemize}
We write $\mathbf{x} \models \Phi$ whenever $(\mathbf{x},0) \models \Phi$.

Given an STL formula $\Phi$, in fact, it is well-known that the satisfaction of  $\Phi$ can be completely determined only by those states within its \emph{horizon}.   
Specifically, we will use notation $\Phi^{[S_\Phi,T_\Phi]}$ to emphasize that the satisfaction of formula $\Phi$ only depends on time horizon  $[S_\Phi,T_\Phi]$, 
where   $S_\Phi$ is the starting instant of $\Phi$ which is the minimum time instant that appears in the formula 
and $T_\Phi$ is the terminal instant of $\Phi$ which is the maximum sum of all nested upper bounds.
For example, for $\Phi = \mathbf{F}_{[2,7]}x^{\mu_1} \wedge  \mathbf{G}_{[3,12]} x^{\mu_2}$, we have $T_\Phi = \max\{7,12\} = 12$ and $S_\Phi = \min\{2, 3\} = 2$.

\section{Problem Formulation}\label{sec:prob}
\subsection{Fragment of STL Formulae}
In this paper, we consider the following restricted but still expressive enough fragments of STL formulae:
\begin{subequations} \label{eq:stl}
	\begin{align}
		\varphi ::= \top \mid \pi^\mu \mid \neg \varphi \mid \varphi_1 \wedge \varphi_2, \qquad \qquad  \\
		\Phi::= \mathbf{F}_{[a,b]} \varphi \mid \mathbf{G}_{[a,b]} \varphi \mid \varphi_1\mathbf{U}_{[a,b]} \varphi_2 \mid \Phi_1 \wedge \Phi_2,
	\end{align}
\end{subequations}
where $\varphi_1, \varphi_2$ are formulae of class $\varphi$, and $\Phi_1, \Phi_2$ are formulae of class $\Phi$.  
Specifically, we only allow the temporal operators be applied once for Boolean formulae. 

Note that, for the standard ``until" operator, 
$\varphi_1\mathbf{U}_{[a,b]} \varphi_2 $ requires that 
$\varphi_1$ holds \emph{from the initial instant} before $\varphi_2$ holds. 
In order to facilitate subsequent expression, we introduce a new temporal operator 
$\mathbf{U}'$ defined by  $(\mathbf{x},k)\models \Phi_1 \mathbf{U}'_{[a,b]} \Phi_2$ 
iff $\exists k' \in [k+a, k+b]$ such that $(\mathbf{x},k') \models \Phi_2$
and  $\forall k'' \in [k+a, k']$, we have $(\mathbf{x},k'') \models \Phi_1$.
Compared with $\mathbf{U}$, the new operator $\mathbf{U}'$ only required that $\varphi_1$ holds \emph{from  instant $a$} before $\varphi_2$ holds. 
Throughout this paper, we will refer ``$\mathbf{U}'$" to as the ``until" operator. 
As illustrated by Figure~\ref{fig:monitor}, our setting is without loss of generality since we can express the standard $\mathbf{U}$ using  $\mathbf{U}'$ by: 
$(\mathbf{x}, k) \models \Phi_1 \mathbf{U}_{[a,b]} \Phi_2$ iff 
\[
	(\mathbf{x}, k) \models (\Phi_1 \mathbf{U}'_{[a,b]} \Phi_2) \wedge (\mathbf{G}_{[0, a]} \Phi_1).
\]

Furthermore, we can always rewrite Boolean formula $\varphi$ in (\ref{eq:stl}a) in terms of the region of states satisfying the formula.
Specifically, for  predicate $\pi^\mu$, its satisfaction region is the solution of inequality $\mu(x) \!\geq\! 0$;  we denote it by set $\mathcal{H}^{\mu}$, i.e., 
\[
\mathcal{H}^{^\mu} = \{x \in \mathcal{X} \mid \mu(x) \geq 0 \}.
\]
Similarly, we have 
$\mathcal{H}^{\neg \varphi} = \mathcal{X} \setminus \mathcal{H}^{\varphi}$ and $\mathcal{H}^{\varphi_1 \wedge \varphi_2} = \mathcal{H}^{\varphi_1} \cap \mathcal{H}^{\varphi_2}$.
Hereafter, instead of using   $\varphi$, 
we will only write it  as $x \in \mathcal{H}^\varphi$ or simply  $x \in \mathcal{H}$  using its satisfaction region. 

Based on the above discussion,  STL formulae $\Phi$ in (\ref{eq:stl}) can be expressed equivalently by:
\begin{align}\label{eq:stl'}
&\Phi  ::=  \\
&  \mathbf{F}_{[a,b]} x \!\in\! \mathcal{H} \mid \mathbf{G}_{[a,b]} x \!\in\! \mathcal{H} \mid   
  x \!\in \mathcal{H}^1 \mathbf{U}'_{[a,b]} x \!\in\! \mathcal{H}^2 \mid \Phi_1 \wedge \Phi_2, \nonumber
\end{align}
where $\mathcal{H}\subseteq \mathbb{R}^n$ is a set of states representing the satisfaction region of a Boolean formula. Finally, we assume that for each temporal operator that appears in $\Phi$, their time intervals have no overlap. In other words, for each time instant $k$, there is at most only one temporal operator applies 
and we denote by $\mathcal{O}_k\in \{\textsf{none},\mathbf{G}, \mathbf{F}, \mathbf{U}'\}$ the unique temporal operator that is effective at instant $k$. 
\begin{remark}
The above assumption  is without loss of generality when two ``always" operators have interval overlap. 
For example, for formula $\Phi=  \mathbf{G}_{[0,2]} x \!\in\! \mathcal{H}_1 \wedge \mathbf{G}_{[1,3]} x \!\in\! \mathcal{H}_2$, we can express it equivalently as $\Phi=  \mathbf{G}_{[0,1]} x \!\in\! \mathcal{H}_1 \wedge
 \mathbf{G}_{[1,2]} x \!\in\! \mathcal{H}_1\cap \mathcal{H}_2
 \wedge  \mathbf{G}_{[2,3]} x \!\in\! \mathcal{H}_2$. 
 However, this assumption is restrictive when  there are overlaps of other operators.  
\end{remark}

In summary, we can write the STL formula under consideration as
\begin{equation}\label{eq:stl-seg}
    \Phi = \bigwedge_{i=1}^{N} \Phi_i^{[a_i, b_i]}, 
\end{equation}
where $N$ denotes the number of sub-formulae 
and each 
$\Phi_i^{[a_i, b_i]}$  is a sub-formula that applies within time interval $[a_i, b_i]$ in the form of 
$\mathbf{G}_{[a_i, b_i]} x \!\in\! \mathcal{H}_i$, $\mathbf{F}_{[a_i, b_i]} x \!\in\! \mathcal{H}_i$ or $x \!\in\! \mathcal{H}_{i}^1 \mathbf{U}'_{[a_i, b_i]} x\! \in\! \mathcal{H}_{i}^2$. 
Without loss of generality, we assume that the horizon of each sub-formulae yields a partition of $[0,T_\Phi]$,
i.e., $[a_1,b_1]\dot{\cup}\cdots \dot{\cup}[a_N,b_N]=[0,T_\Phi]$  with $a_1 = 0$ and $b_N = T_\Phi$. 
This is because if we have $\mathcal{O}_k=\textsf{none}$ for $k\in [a,b]$, 
then we can always add a new formulae $\mathbf{G}_{[a,b]}\top$  or $\mathbf{G}_{[a,b]} x \in \mathcal{X}$ for this interval. 
As a result, hereafter, at instant $k$, the effective temporal operator is  $\mathcal{O}_k\in \{\mathbf{G}, \mathbf{F}, \mathbf{U}'\}$.

\begin{figure}[t]
	\centering
\usetikzlibrary{intersections}
\usetikzlibrary{patterns}
\usetikzlibrary{quotes,angles}
\usetikzlibrary{calc}
\usetikzlibrary{decorations.pathreplacing}

\begin{tikzpicture}
	\foreach \x in {0,1,2,...,15,16}
	\draw (\x*0.5, 2.5-0.03) -- (\x*0.5, 2.5+0.03);
	\foreach \x in {0,1,2,...,15}
	\draw (\x*0.5+0.25, 2.5 - 0.2) node {\tiny \x};
	\draw[very thick] (0,2.5) -- (8,2.5);
	\draw[thick] (0, 2.5-0.1) -- (0, 2.5+0.1);
	\draw[thick] (8, 2.5-0.1) -- (8, 2.5+0.1);
	\draw[thick] (0.5, 2.5-0.1) -- (0.5, 2.5+0.1);
	\draw[thick] (2, 2.5-0.1) -- (2, 2.5+0.1);
	\draw[decorate,decoration={brace,amplitude=5pt}] (0.5, 2.5+0.15) -- (2, 2.5+0.15);
	\draw (1.25, 2.5+0.45) node {\scriptsize $\mathbf{U}$};
	\draw[thick] (3, 2.5-0.1) -- (3, 2.5+0.1);
	\draw[thick] (5, 2.5-0.1) -- (5, 2.5+0.1);
	\draw[decorate,decoration={brace,amplitude=5pt}] (3, 2.5+0.15) -- (5, 2.5+0.15);
	\draw (4, 2.5+0.45) node {\scriptsize $\mathbf{F}$};
	\draw[thick] (6, 2.5-0.1) -- (6, 2.5+0.1);
	\draw[thick] (8, 2.5-0.1) -- (8, 2.5+0.1);
	\draw[decorate,decoration={brace,amplitude=5pt}] (6, 2.5+0.15) -- (8, 2.5+0.15);
	\draw (7, 2.5+0.45) node {\scriptsize $\mathbf{G}$};
	\draw (4,1.9) node {(a)};

	\foreach \x in {0,1,2,...,15,16}
	\draw (\x*0.5, 1-0.03) -- (\x*0.5, 1+0.03);
	\foreach \x in {0,1,2,...,15}
	\draw (\x*0.5+0.25, 1 - 0.2) node {\tiny \x};
	\draw[very thick] (0,1) -- (8,1);
	\draw[thick] (0, 1-0.1) -- (0, 1+0.1);
	\draw[thick] (8, 1-0.1) -- (8, 1+0.1);
	\draw[thick] (0.5, 1-0.1) -- (0.5, 1+0.1);
	\draw[thick] (2, 1-0.1) -- (2, 1+0.1);
	\draw[decorate,decoration={brace,amplitude=5pt}] (0.5, 1+0.15) -- (2, 1+0.15);
	\draw (1.25, 1+0.45) node {\scriptsize $\mathbf{U}'$};
	\draw[decorate,decoration={brace,amplitude=5pt}] (0, 1+0.15) -- (0.5, 1+0.15);
	\draw (0.25, 1+0.45) node {\scriptsize $\mathbf{G}$};
	\draw[thick] (3, 1-0.1) -- (3, 1+0.1);
	\draw[thick] (5, 1-0.1) -- (5, 1+0.1);
	\draw[decorate,decoration={brace,amplitude=5pt}] (3, 1+0.15) -- (5, 1+0.15);
	\draw (4, 1+0.45) node {\scriptsize $\mathbf{F}$};
	\draw[thick] (6, 1-0.1) -- (6, 1+0.1);
	\draw[thick] (8, 1-0.1) -- (8, 1+0.1);
	\draw[decorate,decoration={brace,amplitude=5pt}] (6, 1+0.15) -- (8, 1+0.15);
	\draw (7, 1+0.45) node {\scriptsize $\mathbf{G}$};
	\draw[decorate,decoration={brace,amplitude=5pt}] (2, 1+0.15) -- (3, 1+0.15);
	\draw (2.5, 1+0.45) node {\scriptsize $\mathbf{G}$};
	\draw[decorate,decoration={brace,amplitude=5pt}] (5, 1+0.15) -- (6, 1+0.15);
	\draw (5.5, 1+0.45) node {\scriptsize $\mathbf{G}$};
	\draw (4,0.4) node {(b)};
\end{tikzpicture}
	\caption{Equivalence between operators $\mathbf{U}$ and $\mathbf{U}'$ in online monitoring.}
	\label{fig:monitor}
\end{figure}
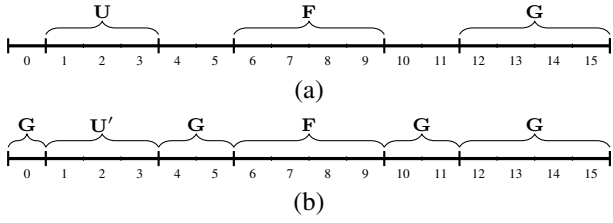

\begin{myexm} 
Let us consider the following  STL formula of form~\eqref{eq:stl}   with horizon  $T_\Phi = 15$
\begin{align}\label{eq:example}
    &\Phi = \\
    &( x \!\in\! \mathcal{H}_1 \mathbf{U}_{[1,3]} x \!\in\! \mathcal{H}_2 ) \wedge ( \mathbf{F}_{[6,9]} x \!\in\! \mathcal{H}_3 ) \wedge ( \mathbf{G}_{[12,15]} x \!\in\! \mathcal{H}_4 ). \nonumber
\end{align}  Its diagram is shown in Figure~\ref{fig:monitor}(a).
Equivalently, we can also write it in the form of~\eqref{eq:stl'}, $\Phi = ( \mathbf{G}_{[0,0]} x \!\in\! \mathcal{H}_1) 
	\wedge 
	( x \!\in\! \mathcal{H}_1 \mathbf{U}'_{[1,3]} x \!\in\! \mathcal{H}_2 ) 
	\wedge
	(\mathbf{G}_{[4,5]} x \!\in\! \mathcal{X})
	\wedge 
	( \mathbf{F}_{[6,9]} x \!\in\! \mathcal{H}_3 ) 
	\wedge
	(\mathbf{G}_{[10,11]} x \!\in\! \mathcal{X})
	\wedge 
	( \mathbf{G}_{[12,15]} x \!\in\! \mathcal{H}_4 )$ as  shown in Figure~\ref{fig:monitor}(b).
\end{myexm}

\subsection{Online Monitoring of STL}
Given a  state sequence $\mathbf{x}$, whose length is equal to or longer than the horizon of $\Phi$, we can always completely determine whether or not $\mathbf{x}\models \Phi$. 
However, during the operation of the system, at each time $k$, we can observe the current state $x_k$, and therefore,  only the partial signal    $\mathbf{x}_{0:k}=x_0x_1\cdots x_k$ (called \emph{prefix}) is available at time instant $k$, and the remaining signals $\mathbf{x}_{k+1: T_{\Phi}}$   (called \emph{suffix}) will only be available in the future. 
We say a prefix signal $\mathbf{x}_{0:k}$ is 
\begin{itemize}
    \item 
    \emph{violated} if   for any control input  $\mathbf{u}_{k: T_{\Phi}-1}$, we have 
$\mathbf{x}_{0:k} \xi_f(x_k, \mathbf{u}_{k:T_{\Phi}}-1) \not\models   \Phi$; 
    \item 
    \emph{feasible} if for some control input  $\mathbf{u}_{k:T_{\Phi}-1}$, we have 
$\mathbf{x}_{0:k} \xi_f(x_k, \mathbf{u}_{k:T_{\Phi}-1})  \models   \Phi$. 
\end{itemize}

Intuitively, a prefix signal is violated if we know for sure in advance that the formula will be violated inevitably. 
For example, for safety specification $\mathbf{G}_{[0,T]}x\!\in\! \mathcal{H}$, once the system reaches a state $x_k\!\notin\! \mathcal{H}$ for $k<T$, we know immediately that the formula is violated. 
Also, if the system is in state $x_k$ from which no solution $\xi_f(x_k,\mathbf{u}_{k:N-1})$ can be found such that each state is in region $\mathcal{H}$, then we can also claim the formula cannot be satisfied anymore. 

Therefore, an online  \emph{monitor} is a function 
\[
\mathcal{M}: \mathcal{X}^* \to \{0,1\}
\] 
that determines the satisfaction of formula based on the partial signal, where $\mathcal{X}^*$ denotes the set of all finite sequences over $\mathcal{X}$, ``$0$" denotes ``feasible" and ``$1$" denotes ``violated".   
Then the online monitoring problem is formulated as follows.  

\begin{myprob}
Given a dynamic system of form (\ref{eq:system}) and an STL formula $\Phi$ as in (\ref{eq:stl-seg}), design 
an online monitor $\mathcal{M}: \mathcal{X}^* \to \{0,1\}$ such that 
for any prefix signal $\mathbf{x}_{0:k}$ where $k \leq T_\Phi$, we have 
$\mathcal{M}(\mathbf{x}_{0:k}) = 1$ iff $\mathbf{x}_{0:k}$ is a violated prefix. 
\end{myprob}

\begin{remark}\label{remark:direct}
We note that, for any prefix signal $\mathbf{x}_{0:k}$,  it is a violated prefix iff we cannot find a sequence of control inputs $\mathbf{u}_{k:T_{\Phi}-1}$  
such that $\mathbf{x}_{0:k} \xi_{f}(x_k, \mathbf{u}_{k:T_{\Phi}-1})  \models \Phi$. 
The existence of such a control sequence can be determined by the binary encoding technique proposed in \cite{raman2014model}.  Therefore, a naive approach for designing  an online monitor is to solve the above constrained satisfaction problem based on $\mathbf{x}_{0:k}$. However, such a direct approach has the following issues 
\begin{itemize}
    \item 
    First, the computations are performed purely online by solving a satisfaction problem, which is computationally very challenging especially for nonlinear systems with long horizon STL formulae.  Consequently, the monitor may not be able to provide an evaluation in time.      
    \item 
    Second, it requires to  store the entire state sequence up to now. 
    It is more desirable if the monitor can just store the satisfaction status of the formula by ``forgetting" those irrelevant information. 
\end{itemize}
\end{remark}
 
Compared with the direct approach discussed in Remark~\ref{remark:direct}, in this paper, we will present an alternative approach by \emph{pre-computing} the set of feasible regions in an offline fashion. Then the pre-computed information will be used online, which ensures timely online evaluations.

\section{Feasible-Set-Based Online Monitoring} \label{sec:online}

\subsection{Subsequent Formulae and Feasible Set}
As we mentioned above, we aim to evaluate the satisfaction of STL formulae of the following form
\begin{equation}\label{eq:stl-seg2}
    \Phi =   \Phi_1^{[a_1, b_1]}\wedge \cdots \wedge \Phi_N^{[a_N, b_N]}. 
\end{equation}
Note that in  the online monitoring problem, once we detect a violated prefix, then the monitoring process is stopped.  
Then for each sub-formula $\Phi_i^{[a_i, b_i]}$, essentially, the monitor needs to determine the following two things within the current monitoring interval $[a_i,b_i]$:
\begin{itemize}
    \item 
    whether or not the current sub-formula has already been satisfied; 
    \item 
    whether or not the system is still able to fulfill the remaining  sub-formulae in the future. 
\end{itemize}

To capture the above issues, we introduce the notion of \emph{subsequent formulae} which is the conjunction of all sub-formulae starting from the current instant.

\begin{mydef}[Subsequent Formulae] 
Given an STL formula $\Phi$ of form~\eqref{eq:stl-seg2}, 
the subsequent formula  after instant $k$ is defined by 
\begin{equation}
\Phi_{sub}^{[k, T_\Phi]}   
=  \Phi_{i_k}^{[k, b_{i_k}]}\wedge\Phi_{i_k+1}^{[a_{i_k+1}, b_{i_k+1}]}
\cdots\wedge\Phi_{N}^{[a_{N}, b_{N}]},
\end{equation}	 
where $i_k = \min \{i \mid k \leq b_i \}$ means that instant $k$ belongs to the monitoring horizon of the $i_k$th sub-formula 
and 
$\Phi_{i_k}^{[k, b_{i_k}]}$ is obtained from 
$\Phi_{i_k}^{[a_{i_k}, b_{i_k}]}$ 
by replacing the starting instant of the temporal operator 
from  $a_{i_k}$ to $k$. \vspace{3pt} 
\end{mydef}

Subsequent formulae will be used when the current sub-formula $\Phi_{i_k}^{[a_{i_k}, b_{i_k}]}$ has not yet been accomplished. 
For temporal operators $\mathbf{F}$ and $\mathbf{U}'$, once the system reaches the target region at instant $k$, the current sub-formula has already been satisfied no matter what happens  between $[k,b_{i_k}]$. 
Then for this case, we just need to monitor the feasibility of the remaining sub-formulae from $\Phi_{i_k+1}^{[a_{i_k+1}, b_{i_k+1}]}$ to 
$\Phi_{N}^{[a_{N}, b_{N}]}$. This is captured by the notion of \emph{exclusive subsequent formulae}.
\begin{mydef}[Exclusive Subsequent Formulae] 
Given an STL formula $\Phi$ of form~\eqref{eq:stl-seg}, 
the exclusive subsequent formula  after instant $k$ is defined by 
\begin{equation}
\hat{\Phi}_{sub}^{[k, T_\Phi]}   
=  \Phi_{i_k+1}^{[a_{i_k+1}, b_{i_k+1}]}\wedge
\cdots\wedge\Phi_{N}^{[a_{N}, b_{N}]}. 
\end{equation}{\vspace{-8pt}}	  
\end{mydef}

Compared with the subsequent formula, the exclusive subsequent formula only excludes the part of interval $[k,b_{i_k}]$, i.e., $\Phi_{sub}^{[k, T_\Phi]} = \Phi_{i_k}^{[k, b_{i_k}]} \wedge  \hat{\Phi}_{sub}^{[k, T_\Phi]}$. 

In order to capture whether or not the (exclusive) subsequent formulae can possibly be fulfilled in the future under the constraint of the system dynamic, we introduce the notion of (exclusive) feasible set. 
\begin{mydef}[Feasible Set]\label{def:feasibleset}
Given an STL formula $\Phi$ of form~\eqref{eq:stl-seg2}, 
the feasible set at instant $k$, denoted by 
$\mathbb{X}_k\subseteq \mathcal{X}$, is the set of states from which there exists a solution that satisfies the subsequent formula $\Phi_{sub}^{[k,T_\Phi]}$, i.e.,
\begin{align}\label{eq:feasibleset}
	\mathbb{X}_k\! =\! \left\{
	x_k \in \mathcal{X} \,\middle\vert\, \!\!\!\!
	\begin{array}{cc}
		\exists \ \mathbf{u}_{k:T_\Phi-1} \in \mathcal{U}^{T_\Phi-k} \\
		\text{ s.t. } x_k \xi_f(x_{k}, \mathbf{u}_{k:T_\Phi-1}) \models \Phi_{sub}^{[k,T_\Phi]}
	\end{array} 
	\right\}.
\end{align}
Similarly,  the exclusive feasible set at instant $k$, denoted by 
$\hat{\mathbb{X}}_k\subseteq \mathcal{X}$ is the set of states from which there exists a solution that satisfies the  exclusive subsequent formula $\hat{\Phi}_{sub}^{[k,T_\Phi]}$.  
\end{mydef}

\begin{myexm}[Cont.]\label{exm:monitor} 
Let us consider the STL formula $\Phi$ in Equation~\eqref{eq:example}. 
For time instant $k=7$, we have $\mathcal{O}_7 = \mathbf{F}$. 
The subsequent formula is
	$\Phi^{[7,15]}_{sub} = ( \mathbf{F}_{[7,9]} x \!\in\! \mathcal{H}_3 ) 
	\wedge
	(\mathbf{G}_{[10,11]} x \!\in\! \mathcal{X})
	\wedge 
	( \mathbf{G}_{[12,15]} x \!\in\! \mathcal{H}_4 )$ 
	and exclusive subsequent formula is 
	$\hat{\Phi}_{sub}^{[7, 15]} = (\mathbf{G}_{[10,11]} x \!\in\! \mathcal{X})
	\wedge 
	( \mathbf{G}_{[12,15]} x \!\in\! \mathcal{H}_4 )$. 
\end{myexm}

In what follows, we will present the main online monitoring algorithm by using feasible sets. 
The  computation of  feasible sets $\mathbb{X}_k$ will be detailed in Section~\ref{sec:offline}.
Our approach is to first discuss the case of  $\mathcal{O}_k = \mathbf{G}$ and the case of $\mathcal{O}_k = \mathbf{F}$ or $\mathbf{U}'$ separately, since the former is a safety property while the latters are reachability properties. 
Then we will combine two cases together to present the main algorithm. 

\subsection{Case I: $\mathcal{O}_k = \mathbf{G}$}
Suppose that at time instant $k$, we have $\mathcal{O}_k = \mathbf{G}$, i.e., the current monitoring sub-formula is of form $\Phi_{i_k}^{[a_{i_k},b_{i_k}]}=\mathbf{G}_{[a_{i_k},b_{i_k}]}x\!\in\! \mathcal{H}_{i_k}$. 
For this case, the online monitor just needs to determine (i) whether or not the system is currently in $\mathcal{H}_{i_k}$;  and 
(ii) whether or not the subsequent formula can be satisfied. This information is completely characterized by the notion of  feasible set $\mathbb{X}_k$, which includes both the requirement for instant $k$ and the future. 
Hence, the monitor just needs to determine whether or not the observed state $x_k$ is in set $\mathbb{X}_k$.  Such a process is summarized by Procedure~1, where $\mathcal{M}_k$ is the abbreviation of   $\mathcal{M}(\mathbf{x}_{0:k})$.

\begin{algorithm}[ht]
    \SetAlgorithmName{Procedure}{}{}
	\caption{Case of  $\mathcal{O}_k = \mathbf{G}$}
	\KwIn{current state $x_k$}
	\If {$x_k \in \mathbb{X}_k $ }
	{
		$\mathcal{M}_k = 0$;
	}
	\Else
	{
		$\mathcal{M}_k = 1$; \\
		\textbf{return} ``\textit{prefix is violated}'';
	}
\end{algorithm}

\subsection{Case II: $\mathcal{O}_k = \mathbf{F}$ or $\mathcal{O}_k = \mathbf{U}'$}
The case of $\mathcal{O}_k = \mathbf{F}$ or $\mathcal{O}_k = \mathbf{U}'$ is different from the case of $\mathcal{O}_k = \mathbf{G}$.  
For example, at instant $7$ in  Example~\ref{exm:monitor}, the sub-formula under monitoring is 
$\Phi_{4}^{[a_4,b_4]}=\mathbf{F}_{[6,9]}x\in \mathcal{H}_3$. 
If we already have that $x_6 \in \mathcal{H}_3$, then this sub-formula is already satisfied. 
Therefore, there is no need to require that states within interval $[7,9]$ are still in $\mathcal{H}_3$ anymore. 
Instead, we just need to focus on whether or not the exclusive sub-formulae $\hat{\Phi}_{sub}^{[k,T_\Phi]}$ can be satisfied.  
 
Based on the above discussion, we propose monitoring process in Procedure~2 for the case of 
$\Phi_i^{[a_i,b_i]}=\mathbf{F}_{[a_i,b_i]}x\!\in\! \mathcal{H}_i$ 
or the case of 
$\Phi_i^{[a_i,b_i]}=x\!\in\! \mathcal{H}_i^1\mathbf{U}'_{[a_i,b_i]}x\!\in\! \mathcal{H}_i^2$. 
Here, for the $i$th formula, we introduce a global Boolean variable $d_i$ to  remember whether or not the $i$th sub-formula has already been satisfied. 
This is determined in line~1, where the variable is set to be $\textsf{true}$ if the target region is visited. 
Depending on the truth value of variable $d_i$, the monitor will take different actions. 
In line~3, when  $d_i = \textsf{false}$,  which means that $\Phi_i^{[a_i, b_i]}$ has not been satisfied, the monitor still needs to monitor the subsequent formula including the remaining part of the current sub-formula,  which is captured by  $\mathbb{X}_k$.
In line~9,  when  $d_i = \textsf{true}$,  which means that $\Phi_i^{[a_i, b_i]}$ has  already been satisfied, the monitor does not need to  monitor  the current sub-formula anymore. 
Instead, it will just focus on the feasibility of the exclusive subsequent formula, which is captured by  $\hat{\mathbb{X}}_k$.

\begin{algorithm}[ht]
    \SetAlgorithmName{Procedure}{}{}
	\caption{Case II: $\mathcal{O}_k = \mathbf{F}$ or $\mathcal{O}_k = \mathbf{U}'$}
	\KwIn{current state $x_k$}
	\If {\emph{[$\mathcal{O}_k \!=\! \mathbf{F}\wedge x_k \in \mathcal{H}_i$]} 
	or \emph{[$ \mathcal{O}_k \!=\! \mathbf{U}'\wedge x_k \in \mathcal{H}_{i}^1 \cap \mathcal{H}_{i}^2$]}}
	{
		set $d_i \gets \textsf{true}$
	}
	
	\If {\emph{$d_i = \textsf{false}$}}
	{
		\If {$x_k \in \mathbb{X}_k $ }
		{
			$\mathcal{M}_k = 0$
		}
		\Else
		{
			$\mathcal{M}_k = 1$ \\
			\textbf{return} ``\textit{prefix is violated}''
		}
	}
	\If {\emph{$d_i = \textsf{true}$}}
	{
		\If {$x_k \in \hat{\mathbb{X}}_k $ }
		{
			$\mathcal{M}_k = 0$
		}
		\Else
		{
			$\mathcal{M}_k = 1$ \\
			\textbf{return} ``\textit{prefix is violated}''
		}
	}

\end{algorithm}

\subsection{Online Monitoring Algorithm}
Based on the above two procedures, we present the complete online monitoring algorithm in Algorithm~1. 
We start from the initial instant $k=0$ and the first sub-formula $i=1$. The satisfaction variables $d_i$ are initialized as \textsf{false} for all sub-formulae $i=1,\dots, N$. 
For for each instant, 
the monitor will obtain new state $x_k$ (Line~4) and use different procedures according to different cases (Lines~5-8).
This process is repeated until the last time instant or 
a violated prefix is detected.

\begin{algorithm}[ht]
	\setcounter{algocf}{0}
	\caption{Online Monitoring Algorithm}
	\KwIn{feasible set $\mathbb{X}_k, \hat{\mathbb{X}}_k$}
	\KwOut{monitoring decision $\mathcal{M}_k$}
	$k \gets 0, i \gets 1 $ \\
	$d_i \gets \textsf{false}, i \in  \{1, \dots, N\} $ \\
	\While{$k \leq T_\Phi$}
	{
		obtain new state $x_k$ \\
		\If {$\mathcal{O}_k = \mathbf{G}$}
		{
			\textbf{Procedure 1}
		}
		\ElseIf {$\mathcal{O}_{k} = \mathbf{F}$ \text{or} $\mathcal{O}_{k} = \mathbf{U}'$}
		{
			\textbf{Procedure 2}
		}
		$k \gets k+1$ \\
		$i \gets \min\{i \mid k \leq b_i\}$ \\ 
	}
	
\end{algorithm}

\begin{remark}
Compared with the direct  approach discussed in Remark~\ref{remark:direct}, the major  advantage of the proposed online monitoring algorithm is  that the online computation burden is very low.  At each time instant, instead of solving a complicated satisfaction problem on-the-fly, our approach just needs to check a set membership. The (exclusive) feasible sets can be computed in an offline fashion and stored in the monitor. Additionally, our algorithm is only based on the current state $x_k$ and do not need to remember the entire trajectory generated by the system.
\end{remark}

\section{Offline Computation  of Feasible Sets}\label{sec:offline}
In this section, we present methods for the computation of (exclusive) feasible sets for each time instant $k$.  
The basis idea is to compute feasible sets recursively in a backwards manner. Specifically,  suppose that we already know the feasible set $\mathbb{X}_{k+1}$, and then we can use $\mathbb{X}_{k+1}$ to compute $\mathbb{X}_{k}$. The specific  computation depends on the current temporal operator that applies, i.e., $\mathcal{O}_k$ is $\mathbf{G}$, $\mathbf{F}$ or $\mathbf{U}'$. In the followings, we will first discuss each case separately and then present the complete algorithm. 

\subsection{Computation of Feasible Sets for $\mathbf{G}$}\label{sec:always}
Suppose that, at time instant $k$, the current monitoring sub-formula is $\Phi_{i_k}^{[a_{i_k},b_{i_k}]}=\mathbf{G}_{[a_{i_k},b_{i_k}]}x\!\in\!\mathcal{H}_{i_k}$. 
If we know that the feasible set for the next instant $\mathbb{X}_{k+1}$ is given, then we can compute the feasible set $\mathbb{X}_{k}$ for the current instant $k$ as the set of states such that 
\begin{itemize}
    \item[(i)] 
    they are in $\mathcal{H}$; and 
    \item[(ii)] 
    they can reach $\mathbb{X}_{k+1}$ in one step under some inputs.
\end{itemize}
This observation is formalized by the \emph{$\mathcal{H}$-one-step set} defined as follows. 

\begin{mydef}[$\mathcal{H}$-One-Step Set]
Let $\mathcal{S}\subseteq \mathcal{X}$ be a set of states representing the ``target region" and $\mathcal{H}\subseteq \mathcal{X}$ be a set of states representing the ``safe region". 
Then the  $\mathcal{H}$-one-step set of $\mathcal{S}$ is defined by 
	\begin{equation}\label{eq:H-one-step}
		\Upsilon_{\mathcal{H}}(\mathcal{S}) = \{x \in \mathcal{H} \mid \exists u \in \mathcal{U}  \text{ s.t. }  f(x, u) \in \mathcal{S}\}. 
	\end{equation}
When $\mathcal{H}=\mathcal{X}$, 
$\Upsilon_{\mathcal{H}}(\mathcal{S})$ is simplified as $\Upsilon(\mathcal{S})$, which is referred to as the one-step set directly.
\end{mydef}

Using the above notation, 
if $\mathcal{O}_k=\mathbf{G}$, then we know that 
\begin{equation}\label{eq:comp-alw}
\mathbb{X}_{k} = \Upsilon_{\mathcal{H}_{i_k}}(\mathbb{X}_{k+1}),
\end{equation}
where $\mathcal{H}_{i_k}$ is the region in which the system should stay during the $i_k$th sub-formula. 

For the sake of convenience, we define 
operator $\Upsilon_{\mathcal{H}}^{(j)}(\mathcal{S})$ inductively by:
\begin{itemize}
    \item
  $\Upsilon^{1}_{\mathcal{H}}(\mathcal{S}) = \Upsilon_{\mathcal{H}}(\mathcal{S})$; and 
  \item 
      $\Upsilon^{(j)}_{\mathcal{H}}(\mathcal{S}) = \Upsilon_{\mathcal{H}}(\Upsilon^{(j-1)}_{\mathcal{H}}(\mathcal{S}))$. 
\end{itemize}
Intuitively, $\Upsilon^{(j)}_{\mathcal{H}}(\mathcal{S})$
is the set of states which can reach region $\mathcal{S}$ in exactly $j$ steps only via states in region $\mathcal{H}$. 

Now, suppose that at instant $k\in [a_{i_k},b_{i_k}]$, we have already computed the feasible set for the starting instant of  next sub-formula, i.e., $\mathbb{X}_{b_{i_k}+1}$. 
Then we have 
\begin{equation}
\mathbb{X}_{k} = \Upsilon_{\mathcal{H}_{i_k}}^{(b_{i_k}+1-k)}(\mathbb{X}_{b_{i_k}+1}).
\end{equation}
Therefore, starting from $\mathbb{X}_{b_{i_k}+1}$, 
all feasible sets within horizon $ [a_{i_k},b_{i_k}]$ can be computed in backwards by applying the $\mathcal{H}$-one-step set operator $\Upsilon_{\mathcal{H}_{i_k}}(\cdot)$ for $b_{i_k}-a_{i_k}+1$ times. 

The exclusive feasible set for each instant $k\in  [a_{i_k},b_{i_k}]$ can be computed analogously. Specifically, 
we just need to replace restricted operator
$\Upsilon_{\mathcal{H}_{i_k}}(\cdot)$
by unrestricted operator $\Upsilon(\cdot)$
and we have
\begin{equation}\label{eq:comp-exl}
\hat{\mathbb{X}}_k=\Upsilon^{(b_{i_k}+1-k)}(\mathbb{X}_{b_{i_k}+1}).
\end{equation}

\subsection{Computation of Feasible Sets for $\mathbf{F}$}\label{sec:eventually}
Now, suppose that, at time instant $k$, the current monitoring sub-formula is $\Phi_{i_k}^{[a_{i_k},b_{i_k}]}=\mathbf{F}_{[a_{i_k},b_{i_k}]}x\!\in\!\mathcal{H}_{i_k}$. 
Also, we assume that the feasible set for the next instant $\mathbb{X}_{k+1}$ has already been computed. Then we know that, when $k \neq b_{i_k}$, a state belongs to feasible set $\mathbb{X}_k$  if \emph{one of} the following two cases holds:
\begin{itemize}
    \item[(i)] 
    it is already in the target region $\mathcal{H}_{i_k}$ for the current monitoring sub-formula, and can continue  to accomplish the exclusive sub-formulae; or
    \item[(ii)]
    it  can reach $\mathbb{X}_{k+1}$ in one step, which means that the task of reaching the target region is postponed to the future instants (no matter it is currently already in $\mathcal{H}_{i_k}$ or not). 
\end{itemize}
Note that, when $k = b_{i_k}$, a state belongs to feasible set $\mathbb{X}_{k}$ only when the first case holds since this is already the last chance to reach target region $\mathcal{H}_{i_k}$.

States satisfying the first case can be characterized as 
$\mathcal{H}_{i_k}\cap \hat{\mathbb{X}}_k$, and  states satisfying the second case can be simply characterized as $\Upsilon(\mathbb{X}_{k+1})$. 
Since a feasible state can be either case, we take the union of these two sets and we have 
\begin{align}\label{eq:comp-eve}
  \mathbb{X}_{k}=\left\{
  \begin{array}{ll}
       (\mathcal{H}_{i_k}\cap \hat{\mathbb{X}}_k) \cup \Upsilon(\mathbb{X}_{k+1}) ,   & \text{if }k \neq b_{i_k}\\
       \mathcal{H}_{i_k}\cap \hat{\mathbb{X}}_k, & \text{if }k = b_{i_k}
  \end{array}
  \right.
\end{align}
where, recalled that, the exclusive feasible set $\hat{\mathbb{X}}_k$ can be computed according to Equation~\eqref{eq:comp-exl}. 
The following result establishes the correctness of the above computation of feasible sets for the case of $\mathbf{F}$. 

\begin{mypro}
Suppose that the current monitoring sub-formula is  $\Phi_{i_k}^{[a_{i_k},b_{i_k}]}=\mathbf{F}_{[a_{i_k},b_{i_k}]}x\!\in\!\mathcal{H}_{i_k}$ and  $\mathbb{X}_{k+1}$ is the feasible set at next time instant. 
Then $\mathbb{X}_k$ computed by Eq.~\eqref{eq:comp-eve} is indeed the feasible set  for the   time instant $k$. 
\end{mypro}

\begin{proof}
When $k = b_{i_k}$, clearly we know that $\mathbb{X}_{k} = \mathcal{H}_{i_k}\cap \hat{\mathbb{X}}_k$ is the feasible set.
For the case of $k \neq b_{i_k}$, since $\mathbb{X}_{k+1}$ is assumed to be the feasible set for instant $k+1$, by definition,  for any state $x_{k+1} \in \mathbb{X}_{k+1}$, there exists $\mathbf{u}_{k+1: T_\Phi-1}$ such that $x_{k+1} \xi_f(x_{k+1}, \mathbf{u}_{k+1: T_\Phi-1}) \models \Phi_{sub}^{[k+1,T_\Phi]}$
where $\Phi_{sub}^{[k+1,T_\Phi]} = \mathbf{F}_{[k+1,b_{i_k}]}x\!\in\!\mathcal{H}_{i_k} \wedge \hat{\Phi}_{sub}^{[k+1, T_\Phi]}$.
At instant $k$, the feasible set $\mathbb{X}_k$ is the set of states from which there exists $\mathbf{u}_{k: T_\Phi-1}$ such that $x_{k} \xi_f(x_{k}, \mathbf{u}_{k: T_\Phi-1}) \models \Phi_{sub}^{[k,T_\Phi]}$, where
\begin{align}
    & \Phi_{sub}^{[k,T_\Phi]} 
    =\mathbf{F}_{[k,b_{i_k}]}x\!\in\!\mathcal{H}_{i_k} \wedge \hat{\Phi}_{sub}^{[k, T_\Phi]} \nonumber \\
    & = (\mathbf{F}_{[k,k]}x\!\in\!\mathcal{H}_{i_k} \wedge \hat{\Phi}_{sub}^{[k, T_\Phi]}) \vee 
    (\mathbf{F}_{[k+1,b_{i_k}]}x\!\in\!\mathcal{H}_{i_k} \wedge \hat{\Phi}_{sub}^{[k, T_\Phi]}) \nonumber \\
    &= (\mathbf{F}_{[k,k]}x\!\in\!\mathcal{H}_{i_k} \wedge \hat{\Phi}_{sub}^{[k, T_\Phi]}) \vee \Phi_{sub}^{[k+1,T_\Phi]}, \nonumber
\end{align}
and $\hat{\Phi}_{sub}^{[k, T_\Phi]}= \hat{\Phi}_{sub}^{[k+1, T_\Phi]}$. 
For the constraint $x_{k} \xi_f(x_{k}, \mathbf{u}_{k: T_\Phi-1}) \models \mathbf{F}_{[k,k]}x\!\in\!\mathcal{H}_{i_k} \wedge \hat{\Phi}_{sub}^{[k, T_\Phi]}$, 
according to the definition of exclusive feasible set, it requires $x_k \!\in\! \mathcal{H}_{i_k}\cap \hat{\mathbb{X}}_k$, while $x_{k} \xi_f(x_{k}, \mathbf{u}_{k: T_\Phi-1}) \models \Phi_{sub}^{[k+1,T_\Phi]}$ requires that $x_k \!\in\! \Upsilon(\mathbb{X}_{k+1})$.
In conclusion, for the case of $k \neq b_{i_k}$,  we have $\mathbb{X}_k = (\mathcal{H}_{i_k}\cap \hat{\mathbb{X}}_k) \cup \Upsilon(\mathbb{X}_{k+1})$.
\end{proof}

\subsection{Computation of Feasible Sets for $\mathbf{U}'$}\label{sec:until}
The case of ``until" is similar to the case of ``eventually". 
Specifically, once again, suppose that, at time instant $k$, the current monitoring sub-formula is $\Phi_{i_k}^{[a_{i_k},b_{i_k}]}=(x \!\in\! \mathcal{H}^1_{i_k}) $ $\mathbf{U}'_{[a_{i_k},b_{i_k}]} (x \!\in\! \mathcal{H}^2_{i_k})$, and 
 the feasible set for the next instant $\mathbb{X}_{k+1}$ is known. 
Then, when $k \neq b_{i_k}$, a state belongs to feasible set $\mathbb{X}_k$  if \emph{one of} the following two cases holds:
\begin{itemize}
    \item[(i)] 
    it is currently in the both regions $\mathcal{H}_{i_k}^1$ and $\mathcal{H}_{i_k}^2$ meaning that the current monitoring sub-formula has already been satisfied, and it still can continue  to accomplish the exclusive sub-formulae; or
    \item[(ii)]
    it can reach $\mathbb{X}_{k+1}$ in one step but only through states in region 
    $\mathcal{H}_{i_k}^1$, which means that $\mathcal{H}_{i_k}^2$ needs to be visited in the future and therefore, the system still needs to stay in $\mathcal{H}_{i_k}^1$. 
\end{itemize}
Also, when $k = b_{i_k}$, a state belongs to feasible set $\mathbb{X}_{k}$ only when  the first case holds since it is already the last time instant for the current sub-formula.

Then similar to Equation~\eqref{eq:comp-eve}, we can also write the feasible set within the horizon of operator $\mathbf{U}'$ as 
\begin{align}\label{eq:comp-unt}
  & \mathbb{X}_{k}= \\
  & \left\{
  \begin{array}{ll}
       (\mathcal{H}_{i_k}^1\cap \mathcal{H}_{i_k}^2\cap \hat{\mathbb{X}}_k) \cup \Upsilon_{\mathcal{H}_{i_k}^1}(\mathbb{X}_{k+1}),   &\text{if } k \neq b_{i_k}\\
       \mathcal{H}_{i_k}^1\cap \mathcal{H}_{i_k}^2\cap \hat{\mathbb{X}}_k. &\text{if } k = b_{i_k}
  \end{array}
  \right. \nonumber 
\end{align}
Also,  the following result establishes the correctness of the above computation of feasible sets for the case of $\mathbf{U}'$. 

\begin{mypro}
Suppose that the current monitoring sub-formula is  $\Phi_{i_k}^{[a_{i_k},b_{i_k}]}=(x \!\in\! \mathcal{H}^1_{i_k}) \mathbf{U}'_{[a_{i_k},b_{i_k}]} (x \!\in\! \mathcal{H}^2_{i_k})$
and  $\mathbb{X}_{k+1}$ is the feasible set at next time instant. 
Then $\mathbb{X}_k$ computed by Eq.~\eqref{eq:comp-unt} is indeed the feasible set  for the   time instant $k$.  
\end{mypro}

\begin{proof}
It is obvious that when $k = b_{i_k}$, $\mathbb{X}_{k} = \mathcal{H}_{i_k}^1\cap \mathcal{H}_{i_k}^2\cap \hat{\mathbb{X}}_k$ is the feasible set.
For the case of $k \neq b_{i_k}$, since $\mathbb{X}_{k+1}$ is assumed to be the feasible set at instant $k+1$,   for any state $x_{k+1} \in \mathbb{X}_{k+1}$, there exists $\mathbf{u}_{k+1: T_\Phi-1}$ such that $x_{k+1} \xi_f(x_{k+1}, \mathbf{u}_{k+1: T_\Phi-1}) \models \Phi_{sub}^{[k+1,T_\Phi]}$
where $\Phi_{sub}^{[k+1,T_\Phi]} = (x \!\in\! \mathcal{H}^1_{i_k}) \mathbf{U}'_{[k+1,b_{i_k}]} (x \!\in\! \mathcal{H}^2_{i_k}) \wedge \hat{\Phi}_{sub}^{[k+1, T_\Phi]}$.
At instant $k$, the feasible set $\mathbb{X}_k$ is the set of states from which there exists $\mathbf{u}_{k: T_\Phi-1}$ such that $x_{k} \xi_f(x_{k}, \mathbf{u}_{k: T_\Phi-1}) \models \Phi_{sub}^{[k,T_\Phi]}$, where
\begin{align}
    & \Phi_{sub}^{[k,T_\Phi]} 
    = (x \!\in\! \mathcal{H}^1_{i_k}) \mathbf{U}'_{[k,b_{i_k}]} (x \!\in\! \mathcal{H}^2_{i_k}) \wedge \hat{\Phi}_{sub}^{[k, T_\Phi]} \nonumber \\
    & = \Big((x \!\in\! \mathcal{H}^1_{i_k}) \mathbf{U}'_{[k,k]} (x \!\in\! \mathcal{H}^2_{i_k}) \wedge \hat{\Phi}_{sub}^{[k, T_\Phi]}\Big) \vee \Phi_{sub}^{[k+1,T_\Phi]}. \nonumber
\end{align}
For the first constraint, it requires $x_k \!\in\! \mathcal{H}_{i_k}^1\cap \mathcal{H}_{i_k}^2\cap \hat{\mathbb{X}}_k$, while the second requires that $x_k \!\in\! \Upsilon_{\mathcal{H}_{i_k}^1}(\mathbb{X}_{k+1})$.
In conclusion, for the case of $k \neq b_{i_k}$, we have  $\mathbb{X}_k = (\mathcal{H}_{i_k}^1\cap \mathcal{H}_{i_k}^2\cap \hat{\mathbb{X}}_k) \cup \Upsilon_{\mathcal{H}_{i_k}^1}(\mathbb{X}_{k+1})$.
\end{proof}

\subsection{Offline Computation Algorithm}\label{sec:offline_alg}
Finally, we summarize the complete procedure for computing all (exclusive) feasible sets within the entire horizon of the formula by combining different cases presented in the previous subsections. The complete process is given by Algorithm~2. 
The iteration starts from the last instant $k=T_\Phi$ for the last sub-formula $i_k = N$  with  a pseudo feasible set $\mathbb{X}_{T_\Phi+1}=\mathbb{R}^n$. 
For each time instant $k$, we compute the exclusive feasible set $\hat{\mathbb{X}}_k$ using the same approach (line~4). The feasible set ${\mathbb{X}}_k$ is computed according to the different cases of $\mathcal{O}_k$ (lines 5-10).  
This process is repeated until iterating to the first time instant $k=0$.
\begin{algorithm}[ht]\label{Algorithm}
	\caption{Computations of All Feasible Sets}
	\KwIn{STL formula $\Phi = \bigwedge_{i=1}^{N} \Phi^{[a_i, b_i]}$}
	\KwOut{all (exclusive) feasible sets $\{\mathbb{X}_k\}$ and $\{\hat{\mathbb{X}}_k\}$}
	$\mathbb{X}_{T_\Phi+1} \gets \mathcal{X}$ \\
	$k \gets T_\Phi, i_k \gets N$ \\
	\While{$k \geq 0$}
	{   
		Compute $\hat{\mathbb{X}}_k$ by Equation~\eqref{eq:comp-exl}\\
	    \If{$\mathcal{O}_{k} = \mathbf{G}$}
	    {
	        Compute $\mathbb{X}_k$ by Equation~\eqref{eq:comp-alw}
	    }
	    \ElseIf{$\mathcal{O}_{k} = \mathbf{F}$}
	    {
	        Compute $\mathbb{X}_k$ by Equation~\eqref{eq:comp-eve}
	    }
	    \ElseIf{$\mathcal{O}_{k} = \mathbf{U}'$}
	    {
	        Compute $\mathbb{X}_k$ by Equation~\eqref{eq:comp-unt}
	    }
		$k \gets k-1$; \\ 
		$i_k \gets \min\{i_k \mid k \leq b_{i_k}\}$;
	}
\end{algorithm}

\subsection{Numerical Computation Considerations}\label{sec:appro}
Finally, we conclude this section by discussing some considerations in the numerical computation of feasible sets.  In order to realize Algorithm~2, the key is to compute the ($\mathcal{H}$-)one-step set  $\Upsilon_{(\mathcal{H})}(\cdot)$. In general, there is no close-form expression for such sets and the computation highly depends on the dynamic of the system.  Particularly, (inner or outer) approximation methodologies have been widely used in practice to achieve the trade off between the  computational  accuracy and complexity.
For example, for linear systems, computation methods for one-step set have been presented subject to polytopic constraints described by linear differential inclusions or for piece-wise affine systems; see, e.g.,  \cite{blanchini1994ultimate, kerrigan2001robust, mayne2001control}. For general nonlinear systems, however, computing the one-step set precisely is much more challenging. For example, \cite{bravo2005computation} proposed a branch and bound algorithm with interval arithmetic approach which provides an inner approximation with a given bound of the error. 

In terms of our monitoring algorithm, it is worth remarking that, if we compute feasible set $\mathbb{X}_k$ by outer-approximations, then  miss-alarms may be possible since we allows some states that are not actually feasible. On the other hand, if we compute feasible set $\mathbb{X}_k$ by inner-approximations, then false-alarms may be possible since the computed feasible sets are conservative. For safety-critical systems,  it is more meaningful to use  inner-approximate of feasible sets to violate miss-alarms. 

Regarding the computation complexity, the overall complexity for computing all feasible sets grows linearly when the horizon of the entire formulae increases. However, for each step in the iteration,  the complexity for computing the one-step sets for constrained systems largely depends on the system model and increases exponentially with the order of the system.  Nevertheless, it is worth  mentioning that computations of feasible sets are purely offline, which does not affect the complexity of the online execution of the monitoring algorithm.

\section{Case Studies for Online Monitoring}\label{sec:case}

In this section, we illustrate our online monitoring algorithm with two examples. 
We show that,  by leveraging the model information of the dynamic system, our model-based approach may provide better monitoring evaluations compared with purely model-free approaches. 

\subsection{Case Study I}

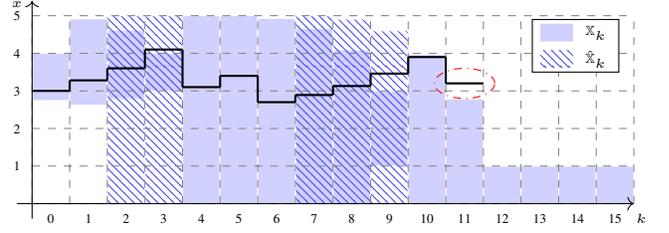
\begin{figure}[t]
	\centering
\usetikzlibrary {patterns.meta}

\def\unit{0.5}

\begin{tikzpicture}
	\filldraw[draw=white, fill=blue!18] (0*\unit, 2.74*\unit) rectangle (1*\unit, 4*\unit);
	\filldraw[draw=white, fill=blue!18] (1*\unit, 2.62*\unit) rectangle (2*\unit, 4.92*\unit);
	\filldraw[draw=white, fill=blue!18, fill opacity = 100] (2*\unit, 2.79*\unit) rectangle (3*\unit, 4.62*\unit);
	\draw[draw = white, pattern=north west lines, pattern color=blue!70] (2*\unit, 0*\unit) rectangle (3*\unit, 5*\unit);
	\filldraw[draw=white, fill=blue!18] (3*\unit, 3*\unit) rectangle (4*\unit, 4*\unit);
	\draw[draw = white, pattern=north west lines, pattern color=blue!70] (3*\unit, 0*\unit) rectangle (4*\unit, 5*\unit);
	\filldraw[draw=white, fill=blue!18] (4*\unit, 0*\unit) rectangle (5*\unit, 5*\unit);
	\filldraw[draw=white, fill=blue!18] (5*\unit, 0*\unit) rectangle (6*\unit, 5*\unit);
	\filldraw[draw=white, fill=blue!18] (6*\unit, 0*\unit) rectangle (7*\unit, 4.93*\unit);
	\filldraw[draw=white, fill=blue!18] (7*\unit, 0*\unit) rectangle (8*\unit, 4.66*\unit);
	\draw[draw = white, pattern=north west lines, pattern color=blue!70] (7*\unit, 0*\unit) rectangle (8*\unit, 5*\unit);
	\filldraw[draw=white, fill=blue!18] (8*\unit, 0*\unit) rectangle (9*\unit, 4.09*\unit);
	\draw[draw = white, pattern=north west lines, pattern color=blue!70] (8*\unit, 0*\unit) rectangle (9*\unit, 4.91*\unit);
	\filldraw[draw=white, fill=blue!18] (9*\unit, 1*\unit) rectangle (10*\unit, 3*\unit);
	\draw[draw = white, pattern=north west lines, pattern color=blue!70] (9*\unit, 0*\unit) rectangle (10*\unit, 4.6*\unit);
	\filldraw[draw=white, fill=blue!18] (10*\unit, 0*\unit) rectangle (11*\unit, 3.97*\unit);
	\filldraw[draw=white, fill=blue!18] (11*\unit, 0*\unit) rectangle (12*\unit, 2.79*\unit);
	\filldraw[draw=white, fill=blue!18] (12*\unit, 0*\unit) rectangle (13*\unit, 1*\unit);
	\filldraw[draw=white, fill=blue!18] (13*\unit, 0*\unit) rectangle (14*\unit, 1*\unit);
	\filldraw[draw=white, fill=blue!18] (14*\unit, 0*\unit) rectangle (15*\unit, 1*\unit);
	\filldraw[draw=white, fill=blue!18] (15*\unit, 0*\unit) rectangle (16*\unit, 1*\unit);

	\draw[step=.5cm,gray,very thin, dashed] (0,0) grid (8,2.5);
	\draw[->] (-0.2,0) -- (8,0);
	\draw[->] (0,-0.2) -- (0,2.7);

	\foreach \x in {0,1,2,...,15}
	\draw (\x*0.5+0.25, - 0.2) node {\tiny \x};
	\draw (8.1, -0.2) node {\tiny $k$};
	\draw (-0.2, 2.65) node {\tiny $x$};
	\foreach \y in {1,2,...,5}
	\draw (-0.2, \y*0.5) node {\tiny \y};

	\draw[thick] (0*\unit, 3*\unit) -- (1*\unit, 3*\unit);
	\draw[thick] (1*\unit, 3.28*\unit) -- (2*\unit, 3.28*\unit);
	\draw[thick] (2*\unit, 3.6*\unit) -- (3*\unit, 3.6*\unit);
	\draw[thick] (3*\unit, 4.1*\unit) -- (4*\unit, 4.1*\unit);
	\draw[thick] (4*\unit, 3.1*\unit) -- (5*\unit, 3.1*\unit);
	\draw[thick] (5*\unit, 3.4*\unit) -- (6*\unit, 3.4*\unit);
	\draw[thick] (6*\unit, 2.7*\unit) -- (7*\unit, 2.7*\unit);
	\draw[thick] (7*\unit, 2.89*\unit) -- (8*\unit, 2.89*\unit);
	\draw[thick] (8*\unit, 3.13*\unit) -- (9*\unit, 3.13*\unit);
	\draw[thick] (9*\unit, 3.46*\unit) -- (10*\unit, 3.46*\unit);
	\draw[thick] (10*\unit, 3.9*\unit) -- (11*\unit, 3.9*\unit);
	\draw[thick] (11*\unit, 3.2*\unit) -- (12*\unit, 3.2*\unit);

	\draw[thick] (1*\unit,3*\unit) -- (1*\unit, 3.28*\unit);
	\draw[thick] (2*\unit,3.28*\unit) -- (2*\unit, 3.6*\unit);
	\draw[thick] (3*\unit,3.6*\unit) -- (3*\unit, 4.1*\unit);
	\draw[thick] (4*\unit,4.1*\unit) -- (4*\unit, 3.1*\unit);
	\draw[thick] (5*\unit,3.1*\unit) -- (5*\unit, 3.4*\unit);
	\draw[thick] (6*\unit,3.4*\unit) -- (6*\unit, 2.7*\unit);
	\draw[thick] (7*\unit,2.7*\unit) -- (7*\unit, 2.89*\unit);
	\draw[thick] (8*\unit,2.89*\unit) -- (8*\unit, 3.13*\unit);
	\draw[thick] (9*\unit,3.13*\unit) -- (9*\unit, 3.46*\unit);
	\draw[thick] (10*\unit,3.46*\unit) -- (10*\unit, 3.9*\unit);
	\draw[thick] (11*\unit,3.9*\unit) -- (11*\unit, 3.2*\unit);

	\filldraw[draw=black, fill=white] (6.65, 1.75) rectangle (7.8, 2.45);
	\filldraw[draw=white, fill=blue!18] (6.75, 2.15) rectangle (7.2, 2.4);
	\filldraw[draw=white, pattern=north west lines, pattern color=blue!70] (6.75, 1.8) rectangle (7.2, 2.05);
	\draw (7.5, 2.275) node {\tiny $\mathbb{X}_k$};
	\draw (7.5, 1.925) node {\tiny $\hat{\mathbb{X}}_k$};

	\draw[red, dash dot] (11.5*\unit, 3.2*\unit) ellipse [x radius=0.4cm, y radius=0.2cm];
	
\end{tikzpicture}
	\caption{A possible signal for  Case  Study I.}
	\label{fig:case}
\end{figure}

As an academic example, let us consider the following one-dimensional discrete-time nonlinear control system 
\[
	x_{k+1} = 0.2 x_k^2 + 0.16 x_k + u_k,
\]
where state constraint is $x_k \!\in\! [0, 5]$
and control input constraint is $u_k \!\in\! [-1,1]$.
The STL formula to monitor is given by 
$\Phi = ( x \!\in\! [0,4] \mathbf{U}_{[1,3]} x \!\in\! [3,5] ) \wedge ( \mathbf{F}_{[6,9]} x \!\in\! [1,3] ) \wedge ( \mathbf{G}_{[12,15]} x \!\in\! [0,1] )$.
Before starting online monitoring, for each time instant, we first compute the (exclusive) feasible set of STL formula $\Phi$ by Algorithm~2 and the results are shown in Fig.~\ref{fig:case}. 
Areas filled with blue and  dots are    the feasible sets $\mathbb{X}_k$ and the exclusive feasible sets $\hat{\mathbb{X}}_k$, respectively. 
For simplicity, we only draw the exclusive feasible sets $\hat{\mathbb{X}}_k$ for the horizon of $\mathbf{U}'$ and $\mathbf{F}$.

During the online monitoring process, the monitor observes the current state at each time and make an evaluation. For example, let us consider a possible state trace generated by the system shown as the black line in Fig.~\ref{fig:case}.
At instant $k=11$,  using the model-free approach, one can only make an inconclusive evaluation since the remaining signal can either satisfy  $\mathbf{G}_{[12,15]} x \in [0,1]$ or not without any constraint. 
However, using our model-based approach, since $x_{11} \notin \mathbb{X}_{11}$, 
we can conclude immediately that the formula will be violated inevitably since there exists no controller under which the STL formula is satisfied.
Therefore, compared with existing model-free algorithms \cite{ho2014online, deshmukh2017robust}, our method can claim the violation of specification in advance at instant 11, while existing algorithms cannot provide a clear violation conclusion.

\subsection{Case Study II}
\begin{figure}[t]
	\centering

\def\unit{0.6}

\begin{tikzpicture}
	\filldraw[draw=black, fill=white] (0, 0) rectangle (10*\unit, 6*\unit);
	\foreach \x in {0,2,4,6,8,10}
	{
	\draw (\x*\unit, -0.2) node {\tiny \x};
	\draw (\x*\unit,-0.08) -- (\x*\unit, 0);
	}
	\foreach \y in {0,2,4,6}
	{
	\draw (-0.2, \y*\unit) node {\tiny \y};
	\draw (-0.08, \y*\unit) -- (0, \y*\unit);
	}

	\filldraw[draw=white, fill=blue!15] (6.1*\unit, 0.2*\unit) rectangle (9.9*\unit, 3.8*\unit);
	\draw (8*\unit, 3.5*\unit) node {$\mathbb{X}_7$};

	\filldraw[draw=white, fill=blue!30] (1*\unit, 2*\unit) rectangle (3*\unit, 4*\unit);
	\draw (1.5*\unit, 3.5*\unit) node {$A_1$};
	\filldraw[draw=white, fill=blue!30] (4*\unit, 4*\unit) rectangle (6*\unit, 6*\unit);
	\draw (4.5*\unit, 5.5*\unit) node {$A_2$};
	\filldraw[draw=white, fill=blue!30] (7*\unit, 1*\unit) rectangle (9*\unit, 3*\unit);
	\draw (7.5*\unit, 2.5*\unit) node {$A_3$};

	\def\A{$A$}
	\coordinate (0) at (1.2*\unit, 1*\unit);
	\coordinate (1) at (2*\unit, 1.8*\unit);
	\coordinate (2) at (1.9*\unit, 2.7*\unit);
	\coordinate (3) at (2.8*\unit, 3.5*\unit);
	\coordinate (4) at (3.7*\unit, 4.3*\unit);
	\coordinate (5) at (4.6*\unit, 4*\unit);
	\coordinate (6) at (5.5*\unit, 4.7*\unit);
	\coordinate (7) at (6.4*\unit, 4*\unit);
	\draw (0) -- (1) -- (2) -- (3) -- (4) -- (5) -- (6) -- (7);
	\foreach \point in {0,1,2,3,4,5,6,7}
	\fill [black,opacity=.7] (\point) circle (2pt);

	\draw[red, dash dot] (6.4*\unit, 4*\unit) circle (0.4*\unit);
	\draw (7.2*\unit, 4*\unit) node {\small{$x_7$}};

\end{tikzpicture}
	\caption{Workspace of the robot in Case Study II.}
	\label{fig:case2}
\end{figure}
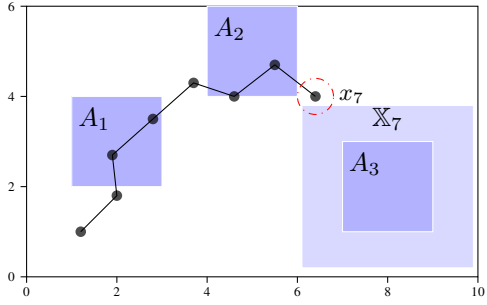
As a practical example, let us consider an autonomous robot whose dynamic model is given as follows 
	\[
	\begin{bmatrix}
		x_{k+1} \\
		y_{k+1} 
	\end{bmatrix} = 
	\begin{bmatrix}
		1 & 0 \\
		0 & 1 
	\end{bmatrix}
	\begin{bmatrix}
		x_{k} \\
		y_{k} 
	\end{bmatrix} + 
	\begin{bmatrix}
		0.9 & 0 \\
		0 & 0.8 
	\end{bmatrix}
	\begin{bmatrix}
		ux_{k} \\
		uy_{k}
	\end{bmatrix},
	\]
where $x_k, y_k, ux_k, uy_k$ are the positions and control inputs in $X$ and $Y$ directions at instant $k$ respectively, and physical constraints are $x_k \!\in\! [0, 10], y_k \!\in\! [0,6]$ and $ ux_k, uy_k \!\in\! [-1,1]$. 

The objective of the robot is to send some raw materials from regions $A_1$ and $A_2$ to region $A_3$ and then stay in region $A_3$ to assemble the machine. By considering the opening time of each region, the robot needs to reach each region within a specified time interval.
Then the task of the robot is  described by the following STL formula 
\[
\Phi = \textbf{F}_{[0,3]}A_1 \wedge \textbf{F}_{[4,6]} A_2 \wedge \textbf{G}_{[8, 10]} A_3,
\]
where $A_1 = (x \!\in\! [1, 3]) \wedge (y \!\in\! [2, 4])$, $A_2 = (x \!\in\! [4, 6]) \wedge (y \!\in\! [4, 6])$ and $A_3 = (x \!\in\! [7, 9]) \wedge (y \!\in\! [1, 3])$.

Consider a trajectory of the robot up to  instant $k=7$  shown in  Figure~\ref{fig:case2}, where the feasible set $\mathbb{X}_7$ computed offline is also depicted; (exclusive) feasible sets for other instants are omitted in the figure for the sake of clarity. 
Then at instant $7$, since the observed state $x_7$ is not in   $\mathbb{X}_7$, the monitor can alarm immediately that the robot has violated the STL task no matter what will happen in the future.
 
\section{Conclusion}\label{sec:con}
In this paper, we proposed a new model-based approach for online monitoring of tasks described by  signal temporal logic formulae. Our algorithm  consists of both offline pre-computation and online monitoring. 
Most of the computation efforts are made for the offline computation characterized by the notion of feasible sets. The offline computed information is used during the online monitoring to provide evaluations in real-time. We showed that the proposed method can evaluate the violation earlier than existing model-free approaches.  Simulation results were provided to illustrate our results.
Note that, in this work, we assume that there is no overlap between the horizon of each temporal operator. In the future, we would like to relax this assumption to further generalize our result.

\bibliographystyle{plain}
\bibliography{STL}   

\end{document}